\documentclass[journal,twoside,web]{ieeecolor}
\usepackage{lcsys}
\usepackage{cite}
\usepackage{amsmath,amssymb,amsfonts}

\usepackage{algorithmic}
\usepackage{graphicx}
\usepackage{textcomp}
\usepackage{xcolor}
\usepackage[ruled, linesnumbered]{algorithm2e}
\usepackage{accents}
\usepackage{mathtools}
\usepackage{lipsum}
\usepackage{enumerate}
\usepackage{hyperref}
\usepackage{adjustbox}

\DeclareSymbolFont{bbold}{U}{bbold}{m}{n}
\DeclareSymbolFontAlphabet{\mathbbold}{bbold}

\usepackage{amsthm}
\newtheorem{theorem}{Theorem}
\newtheorem{proposition}{Proposition}

\newtheorem*{lemma*}{Lemma}
\theoremstyle{definition}
	
	\newtheorem{remark}{Remark}

        \newtheorem{assumption}{Assumption}
	\newtheorem*{example*}{Example}

\newcommand{\change}[1]{#1}

\newcommand{\emptyconc}[2]{\begin{smallmatrix} #1 \\ #2 \end{smallmatrix}}
\newcommand{\smallconc}[2]{\left( \begin{smallmatrix} #1 \\ #2 \end{smallmatrix} \right)}

\newcommand{\geqse}{\geq_{\mathrm{SE}}}
\newcommand{\leqse}{\leq_{\mathrm{SE}}}

\newcommand{\R}{\mathbb{R}}

\newcommand{\IR}{\mathbb{IR}}

\newcommand{\calP}{\mathcal{P}}

\newcommand{\calR}{\mathcal{R}}
\newcommand{\calS}{\mathcal{S}}
\newcommand{\calT}{\mathcal{T}}

\newcommand{\calW}{\mathcal{W}}
\newcommand{\calX}{\mathcal{X}}
\newcommand{\calY}{\mathcal{Y}}

\newcommand{\bfd}{\mathbf{d}}

\newcommand{\bfs}{\mathbf{s}}

\newcommand{\bfw}{\mathbf{w}}

\newcommand{\sfc}{\mathsf{c}}

\newcommand{\sfE}{\mathsf{E}}
\newcommand{\sfF}{\mathsf{F}}
\newcommand{\sfG}{\mathsf{G}}

\newcommand{\sfJ}{\mathsf{J}}

\newcommand{\sfN}{\mathsf{N}}

\newcommand{\sfR}{\mathsf{R}}

\newcommand{\ul}[1]{\underline{#1}}
\newcommand{\ula}{\ul{a}}
\newcommand{\ulb}{\ul{b}}

\newcommand{\uld}{\ul{d}}

\newcommand{\ulv}{\ul{v}}
\newcommand{\ulw}{\ul{w}}
\newcommand{\ulx}{\ul{x}}
\newcommand{\uly}{\ul{y}}
\newcommand{\ulz}{\ul{z}}
\newcommand{\ulA}{\ul{A}}
\newcommand{\ulB}{\ul{B}}

\newcommand{\ol}[1]{\overline{#1}}
\newcommand{\ola}{\ol{a}}
\newcommand{\olb}{\ol{b}}

\newcommand{\old}{\ol{d}}

\newcommand{\olv}{\ol{v}}
\newcommand{\olw}{\ol{w}}
\newcommand{\olx}{\ol{x}}
\newcommand{\oly}{\ol{y}}
\newcommand{\olz}{\ol{z}}
\newcommand{\olA}{\ol{A}}
\newcommand{\olB}{\ol{B}}

\newcommand{\real}{\R}
\newcommand{\overcirc}[1]{\accentset{\circ}{#1}}
\newcommand{\Tgeq}{\calT_{\geq0}}

\begin{document}

\title{Forward Invariance in Neural Network\\Controlled Systems}
\author{Akash Harapanahalli, \IEEEmembership{Student Member, IEEE}, Saber Jafarpour, \IEEEmembership{Member, IEEE}, and Samuel Coogan, \IEEEmembership{Senior Member, IEEE}
\thanks{}
\thanks{Akash Harapanahalli and Samuel Coogan are with the School of Electrical and Computer Engineering at Georgia Institute of Technology, Atlanta, GA, USA. \texttt{\{aharapan,sam.coogan\}@gatech.edu}.}
\thanks{Saber Jafarpour is with the Department of Electrical, Computer, and Energy Engineering, University of Colorado, Boulder. \texttt{saber.jafarpour@colorado.edu}.
}
\thanks{This work was supported in part by the National Science Foundation under grants 1749357 and 2219755 and the Air Force Office of Scientific Research under Grant FA9550-23-1-0303.}
}

\pagestyle{empty}
\maketitle
\thispagestyle{empty}

\begin{abstract}
We present a framework based on interval analysis and monotone systems theory to certify and search for forward invariant sets in nonlinear systems with neural network controllers.
The framework (i) constructs localized first-order inclusion functions for the closed-loop system using Jacobian bounds and existing neural network verification tools; (ii) builds a dynamical embedding system where its evaluation along a single trajectory directly corresponds with a nested family of hyper-rectangles provably converging to an attractive set of the original system; (iii) utilizes linear transformations to build families of nested paralleletopes with the same properties.
The framework is automated in Python using our interval analysis toolbox~\texttt{npinterval}, in conjunction with the symbolic arithmetic toolbox~\texttt{sympy}, demonstrated on an $8$-dimensional leader-follower system.

\end{abstract}


\section{Introduction} \label{sec:introduction}
\IEEEPARstart{L}{earning} enabled components are becoming increasingly prevalent in modern control systems. Their ease of computation and ability to outperform optimization-based approaches make them valuable for in-the-loop usage\change{\cite{SC-etal:18}}. However, neural networks are known to be vulnerable to input perturbations---small changes in their input can yield wildly varying results. 
In safety-critical applications, under uncertainty in the system, it is paramount to verify safe system behavior for an infinite time horizon. Such behaviors are guaranteed through \textit{robust forward invariant sets}, \emph{i.e.}, a set for which a system will never leave under any uncertainty. 

Forward invariant sets are useful for a variety of tasks.  In monitoring, a safety specification extends to infinite time if the system is guaranteed to enter an invariant set. Additionally, for asymptotic behavior of a system, an invariant set can be used as a robustness margin to replace the traditional equilibrium viewpoint in the presence of disturbances.
In designing controllers, one can induce infinite-time safe behavior by ensuring the existence of an invariant set containing all initial states of the system and excluding any unsafe regions. There are several classical techniques in the literature for certifying forward invariant sets, such as Lyapunov-based analysis~\cite{UT-AP-PS:08}, barrier-based methods~\cite{AA-etal:2019}, and set-based approaches~\cite{FB:99}. However, a na\"ive application of these methods generally fails when confronted with high-dimensional and highly nonlinear neural network controllers in-the-loop. 

\paragraph*{Literature review}
The problem of verifying the input-output behavior of standalone neural networks has been studied extensively~\cite{CL-TA-CL-CS-CB-MJK:21}.
There is a growing body of literature studying verification of neural networks applied in feedback loops, which presents unique challenges due to the accumulation of error in closed-loop, \emph{i.e.}, the wrapping effect. 
For example, there are functional approaches such as POLAR~\cite{CH-JF-XC-WL-QZ:22}, JuliaReach~\cite{CS-MF-SG:22}, and ReachMM \change{\cite{SJ-AH-SC:23, SJ-AH-SC:23c}} for nonlinear systems, and linear (resp. semi-definite) programming ReachLP~\cite{ME-GH-CS-JPH:21} (resp. Reach-SDP~\cite{HH-MF-MM-GJP:20}) for linear systems. 
While these methods verify finite time safety, their guarantees do not readily extend to infinite time.
In particular, it is not clear how to adapt these tools to search and certify forward invariant sets of neural network controlled systems.
There are a handful of papers that directly study forward invariance for neural networks in dynamics: In~\cite{AS-RGS:21} a set-based approach is used to study forward invariance of a specific class of control-affine systems with feedforward neural network controllers. In~\cite{HY-PS-MA:22} an ellipsoidal inner-approximation of a region of attraction of neural network controlled system is obtained using Integral Quadratic Constraints (IQCs). In~\cite{HD-BL-LY-MP-RT:21}, a Lyapunov-based approach is used to study robust invariance of control systems modeled by neural networks.  
In~\cite{EB-MG-DP:21}, an adaptive template polytopic approach using MILP verifies RL-based controllers.

\paragraph*{Contributions}

In this letter, we propose a dynamical system approach for systematically finding nested families of robust forward invariant sets for nonlinear systems controlled by neural networks.
Our method uses \emph{localized} first-order inclusion functions to construct an embedding system, which evaluates the inclusion function separately on the edges of a hyper-rectangle.
Our first result is Proposition~\ref{prop:invariance}, which certifies (and fully characterizes for some systems) the forward invariance of a hyper-rectangle through the embedding system's evaluation at a single point.
Our main result is Theorem~\ref{thm:conv}, which describes how a single trajectory of the dynamical embedding system can be used to construct a nested family of invariant and attracting hyper-rectangles.
However, in many applications, hyper-rectangles are not suitable for capturing forward invariant regions, and a simple linear transformation can greatly improve results.
In Proposition~\ref{prop:Tclinclfun}, we carefully construct an accurate localized inclusion function for any linear transformation on the original system, which we use in Theorem~\ref{thm:Tinvsets} to find a nested family of forward invariant paralleletopes.
Finally, we implement the framework in Python, demonstrating its applicability to an $8$-dimensional leader-follower system.\footnote{All code for the numerical experiments can be found at \\ \url{https://github.com/gtfactslab/Harapanahalli_LCSS2024}.} 
\change{In previous work~\cite{SJ-AH-SC:23,SJ-AH-SC:23c}, we consider the online problem of efficiently overapproximating the reachable set of nonlinear learning-enabled systems. In this letter, we consider the offline problem of searching for invariant sets, which greatly benefit from the novel use of localization and state transformations.}
\paragraph*{Notation}
Define the partial order $\leq$ on $\R^n$ as $x\leq y$ if and only if $x_i \leq y_i$ for every $i = 1,\dots,n$. 
For two vectors $\ulx,\olx\in\R^n$ such that $\ulx\leq\olx$, denote the (closed) interval $[\ulx,\olx]=\{x:\ulx\leq x\leq\olx\}$. The set of intervals of $\R^n$ is denoted by $\IR^n$. 
For $[\ula,\ola],[\ulb,\olb]\in\IR$ and $[\ulA,\olA]\in\IR^{m\times p},[\ulB,\olB]\in\IR^{p\times n}$,
\begin{enumerate}
    \item $[\ula,\ola] + [\ulb,\olb] := [\ula + \ulb, \ola + \olb]$ (also on $\IR^n$ element-wise);
    \item $[\ula,\ola]\cdot[\ulb,\olb] := [\min \{\ula\ulb,\ula\olb,\ola\ulb,\ola\olb\}, \max\{\ula\ulb,\ula\olb,\ola\ulb,\ola\olb\}] $;
    \item $([\ulA,\olA][\ulB,\olB])_{i,j} := \sum_{k=1}^{p} [\ulA_{i,k}, \olA_{i,k}]\cdot[\ulB_{k,j},\olB_{k,j}]$.
\end{enumerate}
For two vectors $x,y\in\R^n$ and $i\in\{1,\dots,n\}$, let $x_{i:y}
\in\R^n$ be the vector obtained by replacing the $i$th entry of $x$ with that of $y$, \emph{i.e.},  $(x_{i:y})_j=y_j$ if $i=j$ and otherwise $(x_{i:y})_j=x_j$.

The partial order $\leq$ on $\R^n$ induces the \emph{southeast partial order} $\leqse$ on $\R^{2n}$ as $\big(\emptyconc{x}{\widehat{x}}\big)\leqse\smallconc{y}{\widehat{y}}$ if and only if $x\leq y$ and \change{$\widehat{y}\leq \widehat{x}$.}
Let $\calT^{2n}_{\geq0}=\{\smallconc{x}{\widehat{x}}\in\R^{2n} : x \leq \widehat{x}\}$.
\change{Note that $\Tgeq^{2n}\simeq\IR^n$, and define} $\left[\smallconc{\ulx}{\olx}\right] := [\ulx,\olx]$.
Given a mapping $g$ and a set $\calX\subseteq\operatorname{dom}(g)$, \change{define} the set $g(\calX) := \{g(x) : x\in\calX\}$.

For the nonlinear system $\dot{x} = f(x,w)$ with initial condition $x_0\in\calX_0$ and disturbance $w\in\calW$ \change{for all time}, define the \textit{reachable set} as $
    \calR_{f}(t,\calX_0,\calW) = \left\{
    \begin{aligned}
    \phi&_{f}(t, x_0, \mathbf{w}),\,\forall x_0\in\calX,\\&\mathbf{w}:\change{[0,\infty)}\rightarrow\calW \change{\text{ PW cont.}}
    \end{aligned}
    \right\}
$
where $t\mapsto\phi_f(t,x_0,\bfw)$ is the flow of the system from initial condition $x_0$ at time $0$ under disturbance mapping $\bfw$. A set $\mathcal{X}\subseteq \real^n$ is $\calW$-robustly forward invariant if for every $t\in \real_{\ge 0}$, we have $\calR_{f}(t,\calX,\calW)\subseteq\calX$.   \change{$\mathcal{X}$ is a $\calW$-attracting set with region of attraction $\mathcal{Y}\subseteq \real^n$ if $\calX$ is $\calW$-robustly forward invariant, and for every $x_0\in \calY$, every piecewise continuous $\bfw:\change{[0,\infty)}\to \mathcal{W}$, and every open neighborhood $N\supseteq\calX$, there exists $T^*\ge 0$ such that $\phi_f(t,x_0,\bfw) \in N$, for every $t\ge T^*$.}

\section{Problem Statement} \label{sec:probstatement}

Consider a nonlinear dynamical system of the form 
\begin{align} \label{eq:oldynsys}
    \dot{x} = f(x,u,w),
\end{align}
where $x\in\R^n$ is the state of the system, $u\in\R^p$ is the control input, $w\in\calW\subset\R^q$ is a disturbance, and $f:\R^n\times\R^p\times\R^q\to\R^n$ is a parameterized vector field. 
We assume that the state feedback to the system is defined by a \change{continuously applied} $k$-layer feed-forward neural network controller $N:\R^{n}\to \R^p$:
  \begin{align}\label{eq:NN}
  \begin{aligned}
  \xi^{(0)}&= x, \quad\quad u = N(x) := W^{(k)}\xi^{(k)}(x) +b^{(k)}, \\
   \xi^{(i)} &= \phi^{(i-1)}(W^{(i-1)} \xi^{(i-1)}+ b^{(i-1)}), \quad i\in \{1,\ldots,k\}
  \end{aligned}
    \end{align}
 where $n_i$ is the number of neurons in the $i$th layer, $W^{(i-1)}\in \real^{n_i\times n_{i-1}}$ is the weight matrix of the $i$th layer, $b^{(i-1)}\in \real^{n_i}$ is the bias vector of the $i$th layer, $\xi^{(i)}(y)\in \real^{n_i}$ is the $i$th layer hidden variable, and
$\phi^{(i-1)}:\real^{n_{i}}\to \real^{n_i}$ is the $i$th layer diagonal activation function satisfying $0\le \frac{\phi_j^{(i-1)}(x)-\phi_j^{(i-1)}(y)}{x-y}\le 1$
for every $j\in \{1,\ldots,n_i\}$. A large class of activation functions including ReLU, leaky ReLU, sigmoid, and tanh satisfies this condition (after a possible re-scaling of their co-domain). In feedback with this controller, define the closed-loop neural network controlled system
\begin{align} \label{eq:cldynsys}
    \dot{x} = f^\sfc(x,w) := f(x,N(x),w).
\end{align}
Our goal is to find sets $\calX$ such that, starting inside $\mathcal{X}$, the reachable set of the closed-loop system remains inside $\mathcal{X}$ for all times $t\ge 0$ and for any disturbances in $\calW$, \emph{i.e.}, $\mathcal{W}$-robustly forward invariant sets of the closed-loop system~\eqref{eq:cldynsys}.

\section{Inclusion Functions for Neural Network Controlled Systems} \label{sec:inclfuns}

\subsection{Inclusion Functions} \label{subsec:inclfuns}

Interval analysis aims to provide interval bounds on the output of a function given an interval of possible inputs~\cite{LJ-MK-OD-EW:01}. 
Given a function $f:\R^n\to\R^m$, the function $\sfF=\smallconc{\ul{\sfF}}{\ol{\sfF}}:\Tgeq^{2n}\to\Tgeq^{2m}$ is called an  \textit{inclusion function} for $f$ if
\begin{align}\label{eq:bounds}
    \ul{\sfF}(\ulx,\olx) \leq f(x) \leq \ol{\sfF}(\ulx,\olx),\quad \mbox{for every }x\in[\ulx,\olx],
\end{align}
for every interval $[\ulx,\olx]\subset \R^n$, and is an \textit{$\calS$-localized inclusion function} if the bounds~\eqref{eq:bounds} are valid for every interval $[\ulx,\olx]\subseteq\calS$, in which case we instead write $\sfF_{\calS}$.
Additionally, an inclusion function is
\begin{enumerate}
    \item  \textit{monotone} if $\sfF(\ulx,\olx)\geqse\sfF(\uly,\oly)$, for any $[\ulx,\olx]\subseteq[\uly,\oly]$;
    \item \textit{thin} if for any $x$, we have $\underline{\sfF}(x,x) = f(x) = \overline{\sfF}(x,x)$; 
    \item \textit{minimal} if $\sfF$ returns the tightest possible interval, \emph{i.e.} for each $i\in\{1,\dots,m\}$,
    \begin{align*}
        \ul{\sfF}^\mathrm{min}_i (\ulx,\olx) = \inf_{x\in[\ulx,\olx]} f_i(x), \quad
        \ol{\sfF}^\mathrm{min}_i (\ulx,\olx) = \sup_{x\in[\ulx,\olx]} f_i(x).
    \end{align*}
\end{enumerate}
Given the one-to-one correspondence \change{$\Tgeq^{2n}\simeq\IR^n$,}
an inclusion function is often interpreted as a mapping from intervals to intervals---given an inclusion function $\sfF = \smallconc{\ul{\sfF}}{\ol{\sfF}}:\Tgeq^{2n}\to\Tgeq^{2m}$, we use the notation $[\sfF] = [\ul{\sfF},\ol{\sfF}]:\IR^n\to\IR^m$ to denote the equivalent interval-valued function with interval argument.

In this paper, we focus on two main methods to construct inclusion functions:
\change{(i)} Given a composite function $f = f_1\circ f_2\circ \cdots \circ f_N$, and inclusion functions $\sfF_i$ for $f_i$ for every $i\in\{1,\dots,N\}$, the \textit{natural inclusion function}
\begin{align}\label{eq:natinclfun}
    \sfF^\mathrm{nat} (\ulx,\olx) := (\sfF_1 \circ \sfF_2 \cdots \circ \sfF_N) (\ulx,\olx)
\end{align}
provides a simple but possibly conservative method to build inclusion functions using the inclusion functions of simpler mappings; 
\change{(ii)} Given a differentiable function $f$\change{, an inclusion function $\sfJ_x$ for its Jacobian, and a centering point $\overcirc{x}\in[\ulx,\olx],$ the \emph{Jacobian-based inclusion function}}
\begin{align}\label{eq:jacinclfun}
    [\sfF^\mathrm{jac} (\ulx,\olx)] := [\sfJ_x(\ulx,\olx)]([\ulx,\olx] - \overcirc{x}) + f(\overcirc{x}),
\end{align}
can provide better estimates by bounding the first order Taylor expansion of $f$ \change{around $\overcirc{x}$}.
Both of these inclusion functions are monotone (resp. thin) assuming the inclusion functions used to build them are also monotone (resp. thin).

\change{In previous work~\cite{AH-SJ-SC:23b}, we introduce the open source package \texttt{npinterval} which automates natural inclusion functions in \texttt{numpy}. When used with a symbolic toolbox like \texttt{sympy}, one can construct Jacobian-based inclusion functions.}

\subsection{Localized Closed-Loop Inclusion Functions} \label{subsec:clinclfun}

One of the biggest challenges in neural network controlled system verification is correctly capturing the interactions between the system and the controller. 
For invariance analysis, it is paramount to capture the stabilizing nature of the controller, which can easily be lost with na\"{i}ve overbounding of the input-output interactions between the system and controller. We make the following assumption throughout. 

\begin{assumption}[Local affine bounds of neural network] \label{asm:linearapprox}
    For the neural network~\eqref{eq:NN}, there exists an algorithm that, for any interval $[\ul\xi,\ol\xi]$, produces a local affine bound $(C_{[\ul\xi,\ol\xi]},\uld_{[\ul\xi,\ol\xi]},\old_{[\ul\xi,\ol\xi]})$  such that for every $x\in[\ul\xi,\ol\xi]$,
    \[C_{[\ul\xi,\ol\xi]} x + \uld_{[\ul\xi,\ol\xi]} \leq N(x) \leq C_{[\ul\xi,\ol\xi]} x + \old_{[\ul\xi,\ol\xi]}.\]
\end{assumption}
Many off-the-shelf neural network verification frameworks can produce the linear estimates required in Assumption~\ref{asm:linearapprox}, and in particular, we focus on CROWN~\cite{HZ-etal:18}. For ReLU and otherwise piecewise linear networks, one can setup a mixed integer linear program similar to~\cite{VT-KX-RT:19}, which is tractable for small-sized networks.
\change{Frameworks like} {\verb|auto_LiRPA|}\change{\cite{xu2020automatic} operate on general computational graphs, and thus satisfy this assumption for a wide variety neural network architectures, \emph{e.g.}, residual neural networks, recurrent neural networks, and convolutional neural networks.}

The bounds from Assumption \ref{asm:linearapprox} can be used to construct a $[\ul\xi,\ol\xi]$-localized inclusion function for $N(x)$:
\begin{align} \label{eq:NNinclfun}
\begin{aligned}
    \ul{\sfN}_{[\ul\xi,\ol\xi]}(\ulx,\olx) &= C^+_{[\ul\xi,\ol\xi]} \ulx + C^-_{[\ul\xi,\ol\xi]} \olx + \uld_{[\ul\xi,\ol\xi]},\\  
    \ol{\sfN}_{[\ul\xi,\ol\xi]}(\ulx,\olx) &= C^-_{[\ul\xi,\ol\xi]} \ulx + C^+_{[\ul\xi,\ol\xi]} \olx + \old_{[\ul\xi,\ol\xi]},
\end{aligned}
\end{align}
where $(C^+)_{i,j} = \max(C_{i,j},0)$, and $C^- = C - C^+$. 

Using the localized first-order bounds of the neural network, we propose a general framework for constructing closed-loop inclusion functions for neural network-controlled systems that capture the first-order stabilizing effects of the controller. First, assuming $f$ is differentiable, with \change{inclusion functions $\sfJ_x,\sfJ_u,\sfJ_w$ for the Jacobians $D_xf,D_uf,D_wf$}, one can construct a closed-loop Jacobian-based inclusion function $\sfF^\sfc$. \change{Given an interval $[\ulz,\olz]$, with $\sfJ_x,\sfJ_u,\sfJ_w$} evaluated on the input $(\ulz,\olz,\ul{\sfN}_{[\ulz,\olz]}(\ulz,\olz),\ol{\sfN}_{[\ulz,\olz]}(\ulz,\olz),\ulw,\olw)$, define 
\begin{align} \label{eq:clinclfun}
\begin{aligned}
    [\sfF^\sfc_{[\ulz,\olz]} &(\ulx,\olx,\ulw,\olw)]  = ([\sfJ_x] + [\sfJ_u] C_{[\ulx,\olx]}) [\ulx,\olx] \\ &+ [\sfJ_u] [\uld_{[\ulx,\olx]},\old_{[\ulx,\olx]}] + [\sfR_{[\ulz,\olz]}(\ulw,\olw)],
\end{aligned}
\end{align}
where $\overcirc{x}\in[\ulx,\olx]\subseteq[\ulz,\olz]$, $\overcirc{u}=N(\overcirc{x})$, $\overcirc{w}\in[\ulw,\olw]$, and $[\sfR_{[\ulz,\olz]}(\ulw,\olw)] := -[\sfJ_x]\overcirc{x} - [\sfJ_u]\overcirc{u}  + [\sfJ_w]([\ulw,\olw] - \overcirc{w}) + f(\overcirc{x},\overcirc{u},\overcirc{w})$. 
Proposition~\ref{prop:Tclinclfun} provides a more general result proving~\eqref{eq:clinclfun} is a $[\ulz,\olz]$-localized inclusion function for $f^\sfc$ ($T=I$).

In the case that $f$ is not differentiable, or finding an inclusion function for its Jacobian is difficult, as long as a (monotone) inclusion function $\sfF$ for the open-loop system $f$ is known, a (monotone) closed-loop inclusion function for $f^\sfc$ from~\eqref{eq:cldynsys} can be constructed using the natural inclusion approach in~\eqref{eq:natinclfun} with $\sfN$ from~\eqref{eq:NNinclfun}
\begin{align} \label{eq:clinterbased}
    \sfF^\sfc(\ulx,\olx,\ulw,\olw) = \sfF(\ulx,\olx,\ul{\sfN}_{[\ulx,\olx]}(\ulx,\olx),\ol{\sfN}_{[\ulx,\olx]}(\ulx,\olx),\ulw,\olw).
\end{align}
\vspace{-1.5em}
\change{ \begin{remark}[Interval observer]
    In the case that the system state is not perfectly known
    but rather the output of an interval observer, one can modify~\eqref{eq:cldynsys} to $\dot{x} = f(x,N(x + v),w)$, where $v\in[\ulv,\olv]$ is the interval observer uncertainty. One can incorporate this into either inclusion function by bloating the localization of calls to Assumption~\ref{asm:linearapprox}. For~\eqref{eq:clinclfun}, as long as $[\ulx+\ulv,\olx+\olv]\subseteq[\ulz,\olz]$, replace $(C_{[\ulx,\olx]},\uld_{[\ulx,\olx]},\old_{[\ulx,\olx]})$ with $(C_{[\ulx+\ulv,\olx+\olv]},\uld_{[\ulx+\ulv,\olx+\olv]},\old_{[\ulx+\ulv,\olx+\olv]})$. 
    For~\eqref{eq:clinterbased}, replace $\sfN_{[\ulx,\olx]}$ with $\sfN_{[\ulx+\ulv,\olx+\olv]}$.
\end{remark} }

\section{A Dynamical Approach to Set Invariance} \label{sec:embeddingsystem}

Using the closed-loop inclusion functions developed in Section~\ref{sec:inclfuns}, we embed the uncertain dynamical system~\eqref{eq:cldynsys} into a larger certain system that enables computationally tractable approaches to \change{verify and compute} families of invariant sets. 
Consider the closed-loop system~\eqref{eq:cldynsys} with an $\calS$-localized inclusion function $\sfF^\sfc_{\calS}:\Tgeq^{2n}\times\Tgeq^{2q}\to\Tgeq^{2n}$ for $f^\sfc$ constructed via, \emph{e.g.}, \eqref{eq:clinclfun} or \eqref{eq:clinterbased}, with the disturbance set $\calW\subseteq[\ulw,\olw]$. Then $\sfF^\sfc_{\calS}$ \textit{induces} an \textit{embedding system} for~\eqref{eq:cldynsys} with state $\smallconc{\ulx}{\olx}\in \Tgeq^{2n}$ and dynamics defined by
\begin{align}\label{eq:embsys}
\begin{aligned}
    \dot{\ulx}_i = \Big(\ul{\sfE}_{\calS}^{\change{\sfc}}(\ulx,\olx,\ulw,\olw)\Big)_i := \Big(\ul{\sfF}^\sfc_{\calS}(\ulx,\olx_{i:\ulx},\ulw,\olw)\Big)_i,\\
    \dot{\olx}_i = \Big(\ol{\sfE}_{\calS}^{\change{\sfc}}(\ulx,\olx,\ulw,\olw)\Big)_i := \Big(\ol{\sfF}^\sfc_{\calS}(\ulx_{i:\olx},\olx,\ulw,\olw)\Big)_i,
\end{aligned}
\end{align}
where $\sfE_{\calS}^{\change{\sfc}}:\Tgeq^{2n}\times\Tgeq^{2q}\to\R^{2n}$.
One of the key features of the embedding system, which evolves on $\Tgeq^{2n}$, is that the inclusion function is evaluated separately on each face of the hyper-rectangle $[\ulx,\olx]$, represented by $[\ulx,\olx_{i:\ulx}]$ and $[\ulx_{i:\olx},\olx]$ for each $i\in\{1,\dots,n\}$.
In Proposition~\ref{prop:invariance}, this meshes nicely with Nagumo's Theorem~\cite{FB:99}, which allows us to guarantee forward invariance by checking the boundary of the hyper-rectangle through one evaluation of the embedding system~\eqref{eq:embsys}. 

\begin{proposition}[Forward invariance in hyper-rectangles]\label{prop:invariance}
Consider the neural network controlled system~\eqref{eq:cldynsys} with the disturbance set $\mathcal{W}=[\ulw,\olw]$ \change{and initial condition $x_0\in[\ulx^\star,\olx^\star]$. Given a set $\calS\supseteq[\ulx^\star,\olx^\star]$, let} $\sfF^{\sfc}_\calS$ be a $\calS$-localized inclusion function for $f^{\sfc}$, \emph{e.g.}~\eqref{eq:clinclfun} or~\eqref{eq:clinterbased}, and $\sfE_\calS^{\change{\sfc}}$ be the embedding system induced by $\sfF^{\sfc}_\calS$. %
If
\begin{align}\label{eq:verygoodcondition}
     \sfE_\calS^{\change{\sfc}} (\ulx^\star,\olx^\star,\ulw,\olw) \geqse 0,
\end{align}
then $[\ulx^\star,\olx^\star]$ is a $[\ulw,\olw]$-robustly forward invariant set. Moreover, if $\sfF^\sfc_\calS$ is the minimal inclusion function \change{of $f^\sfc$}, the condition~\eqref{eq:verygoodcondition} is also necessary for $[\ulx^\star,\olx^\star]$ to be a $[\ulw,\olw]$-robustly forward invariant set.
\end{proposition}
\begin{proof}
    For brevity, since $[\ulx^\star,\olx^\star]\subseteq\calS$, we drop $\calS$ from the notation.
    Consider the set $[\ulx^\star,\olx^\star]$, and suppose that $\sfE^{\change{\sfc}}(\ulx^\star,\olx^\star,\ulw,\olw) \geqse 0$. 
    Therefore, for every $i\in\{1,\dots,n\}$, 
    \begin{align*}
        0 \leq \ul{\sfF}^\sfc_i(\ulx^\star,\olx^\star_{i:\ulx^\star},\ulw,\olw) \leq \inf_{x\in[\ulx^\star,\olx^\star_{i:\ulx^\star}],w\in[\ulw,\olw]}f^\sfc(x,w). 
    \end{align*} 
    This implies that $f^\sfc(x,w)\geq 0$ for every $x$ on the $i$-th lower face of the hyperrectangle $[\ulx^\star,\olx^\star_{i:\ulx^\star}]$. Similarly, $f^\sfc(x,w)\leq 0$ on the $i$-th upper face of the hyperrectangle $[\ulx^\star_{i:\olx^\star},\olx^\star]$.
    Since this holds for every $i\in \{1,\ldots,n\}$, by Nagumo's theorem~\cite[Theorem 3.1]{FB:99}, the closed set $[\ulx^\star, \olx^\star]$ is forward invariant since for every point $x$ along its boundary $\bigcup_i ([\ulx^\star,\olx^\star_{i:\ulx^\star}] \cup [\ulx^\star_{i:\olx^\star},\olx^\star])$, the vector field $f^\sfc(x,w)$ points into the set, for every $w\in [\ulw,\olw]$.
    Now, suppose that $\sfE$ is the embedding system induced by the minimal inclusion function $\sfF^\mathrm{min}$ for $f^\sfc$, and $[\ulx^\star,\olx^\star]$ is a hyper-rectangle such that $\sfE(\ulx^\star,\olx^\star,\ulw,\olw) \not\geq_\text{SE} 0$. Then there exists $i\in \{1,\ldots,n\}$ such that either
    \begin{align*} \ul{\sfF}^\mathrm{min}_i(\ulx^\star,\olx^\star_{i:\ulx^\star},\ulw,\olw) &= \inf_{x\in[\ulx^\star,\olx^\star_{i:\ulx^\star}],w\in[\ulw,\olw]}f(x,w)<0, \text{ or} \\
        \ol{\sfF}^\mathrm{min}_i(\ulx^\star_{i:\olx^\star},\olx^\star,\ulw,\olw) &= \sup_{x\in[\ulx^\star_{i:\olx^\star},\olx^\star],w\in[\ulw,\olw]}f(x,w)>0.
    \end{align*}
    If the first case holds, then there exists $x'\in[\ulx^\star,\olx^\star_{i:\ulx}],w \in [\ulw,\olw]$ such that $f^\sfc(x',w)<0$ along the $i$-th lower face of the hyper-rectangle. 
    If the second case holds, then there exists $x'\in[\ulx^\star_{i:\olx},\olx^\star],w \in [\ulw,\olw]$ such that $f^\sfc(x',w)>0$ along the $i$-th upper face of the hyper-rectangle.
    Thus, by Nagumo's theorem, the set $[\ulx^\star,\olx^\star]$ is not $[\ulw,\olw]$-robustly forward invariant, as there exists a point along its boundary such that the vector field $f^\sfc$ points outside the set.
\end{proof}
\begin{remark}[Linear systems \change{with piecewise linear activations}]%
    For the special case of the linear system $\dot{x} = Ax + Bu + Dw$ controlled by a neural network $u=N(x)$ with piecewise linear activations, \emph{e.g.} ReLU or Leaky ReLU, %
    one can compute the minimal inclusion function using a Mixed Integer Linear Program (MILP) similar to~\cite{VT-KX-RT:19}. For $i\in\{1,\dots,n\}$,
    \begin{align*}
    \ul{\sfF}^\mathrm{min}_i(\ulx,\olx,\ulw,\olw) &= \min_{x\in[\ulx,\olx],w\in[\ulw,\olw]} (Ax + BN(x) + Dw)_i, \\
    \ol{\sfF}^\mathrm{min}_i(\ulx,\olx,\ulw,\olw) &= \max_{x\in[\ulx,\olx],w\in[\ulw,\olw]} (Ax + BN(x) + Dw)_i.
    \end{align*}
\end{remark}

The next Theorem shows how monotonicity of the \change{embedding dynamics define} a family of nested robustly forward invariant sets using the condition from Proposition~\ref{prop:invariance}.

\begin{theorem}[A nested family of invariant sets] \label{thm:conv}
Consider the neural network controlled system~\eqref{eq:cldynsys} with the disturbance set $\mathcal{W}=[\ulw,\olw]$ \change{and initial condition $x_0\in[\ulx_0,\olx_0]$. Given a set $\calS\supseteq[\ulx_0,\olx_0]$, let} $\sfF^{\sfc}_\calS$ be a $\calS$-localized monotone inclusion function for $f^{\sfc}$, \emph{e.g.}~\eqref{eq:clinclfun} or~\eqref{eq:clinterbased}, and $\sfE_\calS^{\change{\sfc}}$ be the embedding system induced by $\sfF^{\sfc}_\calS$. %
If 
\[\sfE_\calS^{\change{\sfc}}(\ulx_0,\olx_0,\ulw,\olw) \geqse 0,\]
\change{then} 
\begin{enumerate}[i)]
    \item\label{p2:conv} $[\ulx(t),\olx(t)]$ is a $[\ulw,\olw]$-robustly forward invariant set for the system~\eqref{eq:cldynsys} for every $t\ge 0$, and for every $t\leq \tau$, $[\ulx(\tau),\olx(\tau)] \subseteq [\ulx(t),\olx(t)]$; and
    \item\label{p3:conv} $\lim_{t\to\infty} \smallconc{\ulx(t)}{\olx(t)} = \smallconc{\ulx^\star}{\olx^\star}$, where $\smallconc{\ulx^\star}{\olx^\star} \in \mathcal{T}^{2n}_{\ge 0}$ is  an equilibrium point of the embedding system~\eqref{eq:embsys} and $[\ulx^\star,\olx^\star]$ is a $[\ulw,\olw]$-attracting set for the system~\eqref{eq:cldynsys} with region of attraction $[\ulx_0,\olx_0]$,
\end{enumerate}
\change{where $t\mapsto \smallconc{\ulx(t)}{\olx(t)}$ is the trajectory of~\eqref{eq:embsys} from $\Big(\emptyconc{\ulx_0}{\olx_0}\Big)$.}
\end{theorem}
\begin{proof}
\change{(Monotonicity of $\sfE_\calS^\sfc$ dynamics).\ 
Consider any two points $\Big(\emptyconc{\ulx}{\olx}\Big),\smallconc{\ulx'}{\olx'}\in\Tgeq^{2n}$, such that $\Big(\emptyconc{\ulx}{\olx}\Big) \leqse \smallconc{\ulx'}{\olx'}$. This implies that $\Big(\emptyconc{\ulx}{\olx_{i:\ulx}}\Big) \leqse \Big(\emptyconc{\ulx'}{\olx'_{i:\ulx'}}\Big)$ and $\Big(\emptyconc{\ulx_{i:\olx}}{\olx}\Big) \leqse \smallconc{\ulx'_{i:\olx'}}{\olx'}$. 
Since $\sfF^\sfc_\calS$ is a monotone inclusion function,
\begin{align*}
    \big(\ul{\sfF}^\sfc_{\calS}(\ulx,\olx_{i:\ulx},\ulw,\olw)\big)_i & \leq
    \big(\ul{\sfF}^\sfc_{\calS}(\ulx',\olx'_{i:\ulx'},\ulw,\olw)\big)_i, \\
    \big(\ol{\sfF}^\sfc_{\calS}(\ulx_{i:\ulx},\olx,\ulw,\olw)\big)_i & \geq
    \big(\ol{\sfF}^\sfc_{\calS}(\ulx'_{i:\ulx'},\olx',\ulw,\olw)\big)_i,
\end{align*}}
for every $i\in \{1,\ldots,n\}$. This implies that the embedding system $\sfE^\sfc_\calS$ is monotone w.r.t. the southeast partial order $\leqse$~\cite{DA-EDS:03, GAE-HLS-EDS:06}. 
\change{(Part (i)).} \ 
Now using~\cite[Proposition 2.1]{HLS:95}, the set $\calP_+=\left\{\smallconc{\ulx}{\olx} : \smallconc{\ulx}{\olx} \geqse 0 \right\}$ is a forward invariant set for $\sfE^\sfc_\calS$. Since $\Big(\emptyconc{\ulx_0}{\olx_0}\Big)\in\calP_+$, forward invariance implies $\smallconc{\ulx(t)}{\olx(t)}\in\calP_+$ for every $t\geq 0$. Therefore, using Proposition~\ref{prop:invariance}, $[\ulx(t),\olx(t)]$ is forward invariant for the closed-loop system~\eqref{eq:cldynsys} for every $t\geq 0$. Additionally, the curve $t\mapsto \smallconc{\ulx(t)}{\olx(t)}$ is nondecreasing with respect to the partial order $\leqse$~\cite[Proposition 2.1]{HLS:95}. This means that for every $t\le \tau$, $\smallconc{\ulx(t)}{\olx(t)} \leqse \smallconc{\ulx(\tau)}{\olx(\tau)}$, which implies that $[\ulx(\tau),\olx(\tau)]\subseteq [\ulx(t),\olx(t)]$.
\change{(Part (ii)).}\ 
Since $t\mapsto \smallconc{\ulx(t)}{\olx(t)}$ is nondecreasing w.r.t. $\leqse$, for every $i\in \{1,\ldots,n\}$, the curves $t\mapsto \ulx_i(t)$ (resp. $t\mapsto \olx_i(t)$) are nondecreasing (resp. nonincreasing) w.r.t. $\leq$ and bounded on $\real$. This implies that there exists $\ulx^*_i$ and $\olx_i^*$ such that $\lim_{t\to\infty} \ulx_i(t) = \ulx_i^*$ and $\lim_{t\to\infty} \olx_i(t) = \olx_i^*$. By defining the vector $\ulx^* = (\ulx^*_1,\ldots,\ulx^*_n)^{\top}$ and $\olx^* = (\olx^*_1,\ldots,\olx^*_n)^{\top}$, we get $\lim_{t\to\infty} \smallconc{\ulx(t)}{\olx(t)} = \smallconc{\ulx^\star}{\olx^\star}$. Moreover, since $\smallconc{\ulx(t)}{\olx(t)}\in \mathcal{T}^{2n}_{\ge 0}$ for every $t\ge 0$ and is a continuous curve in time $t$, we get $\smallconc{\ulx^\star}{\olx^\star}\in \mathcal{T}^{2n}_{\ge 0}$.
Finally, $\calR_{f^\sfc}(t,[\ulx_0,\olx_0],[\ulw,\olw])\subseteq [\ulx(t),\olx(t)]$ for every $t\ge 0$~\cite[Proposition 5]{SJ-AH-SC:23c}. Thus, every trajectory of the system starting from $[\ulx_0,\olx_0]$ converges to $[\ulx^\star,\olx^{\star}]$.
\end{proof}

After finding one invariant set using Proposition~\ref{prop:invariance}, Theorem~\ref{thm:conv} obtains a nested family of invariant sets guaranteed to converge to some 
$[\ulx^\star,\olx^\star]$, obtained by evolving the embedding system forwards in time. This is the smallest invariant set in the family, and is a $[\ulw,\olw]$-attractive set with region of attraction $[\ulx_0,\olx_0]$. Note that the embedding system can also be evolved backwards in time while $[\ulx(t),\olx(t)]\geqse 0$ and $[\ulx(t),\olx(t)]\subseteq\calS$ to expand the invariant sets. 
Additionally, Proposition~\ref{prop:invariance} and Theorem~\ref{thm:conv} do not require thin inclusion functions, \change{generalizing} existing decomposition-based approaches for \change{finding} hyper-rectangular invariant sets~\cite{MA-SC:20}.

\section{Paralleletope Invariant Sets} \label{sec:paralleletope}

Theorem~\ref{thm:conv} can be used to verify and search for hyper-rectangular invariant sets using the embedding system. %
In this section, we extend our framework to characterize a more general class of invariant paralleletopes. For some invertible matrix $T\in \R^{n\times n}$, define the $T$-transformed system 
\begin{align} \label{eq:Tclsys}
\begin{aligned}
    y &:= Tx := \Phi(x) \\
    \dot{y} &= g^\sfc(y,w) = Tf(T^{-1}y, N'(y), w),
\end{aligned}
\end{align}
where $N'(y) := N(T^{-1}y)$ is the neural network from~\eqref{eq:NN}, with an extra initial layer with weight matrix $T^{-1}$ and linear activation $\phi(x) = x$.
There is a one-to-one correspondence between the transformed system~\eqref{eq:Tclsys} and the original system~\eqref{eq:cldynsys}, in the sense that every trajectory $t\mapsto y(t)$ of~\eqref{eq:Tclsys} uniquely corresponds with the trajectory $t\mapsto\Phi^{-1}(y(t))$ of~\eqref{eq:cldynsys}. Given an interval $[\uly,\oly]$, the set $\Phi^{-1}([\uly,\oly]) = \{T^{-1}y : y\in[\uly,\oly]\}$ defines a paralleletope in standard coordinates.
We construct a localized closed-loop Jacobian-based inclusion function $\sfG^\sfc_{\Phi([\ulz,\olz])}$ \change{for $g^\sfc$ as follows. Given an interval $[\ulz,\olz]$ in standard coordinates, with inclusion functions $\sfJ_x,\sfJ_u,\sfJ_w$ for the Jacobians $D_xf,D_uf,D_wf$ of the original system evaluated on the input $(\ulz,\olz,\ul{\sfN}_{[\ulz,\olz]}(\ulz,\olz),\ol{\sfN}_{[\ulz,\olz]}(\ulz,\olz),\ulw,\olw)$, define}
\begin{align} 
    [\sfG^\sfc_{\Phi([\ulz,\olz])} &(\uly,\oly,\ulw,\olw)]  = T([\sfJ_x] + [\sfJ_u] (C'_{[\uly,\oly]}T))T^{-1}[\uly,\oly] \nonumber \\ 
    \label{eq:Tclinclfun}
    &+ T[\sfJ_u] [\uld'_{[\uly,\oly]},\old'_{[\uly,\oly]}] + T[\sfR_{[\ulz,\olz]}(\ulw,\olw)],
\end{align}
where $(C',\uld',\old')$ are from Assumption~\ref{asm:linearapprox} evaluated over $[\uly,\oly]$ on the transformed neural network $N'(y) = N(T^{-1}y)$, $\overcirc{x}\in\change{\Phi^{-1}([\uly,\oly])\subseteq[\ulz,\olz]}$, $\overcirc{u}\change{\in\sfN_{[\ulz,\olz]}(\ulz,\olz)}$, $\overcirc{w}\in[\ulw,\olw]$\change{, and $[\sfR_{[\ulz,\olz]}(\ulw,\olw)] = -[\sfJ_x]\overcirc{x} - [\sfJ_u]\overcirc{u}  + [\sfJ_w]([\ulw,\olw] - \overcirc{w}) + f(\overcirc{x},\overcirc{u},\overcirc{w})$.}

\begin{proposition} \label{prop:Tclinclfun}
    Consider the neural network controlled system~\eqref{eq:cldynsys} 
    and let $T\in\R^{n\times n}$ be an invertible matrix transforming the system into~\eqref{eq:Tclsys}. Then~\eqref{eq:Tclinclfun} is a $\Phi([\ulz,\olz])$-localized (monotone) inclusion function for $g^\sfc$.
\end{proposition}
\begin{proof}
\change{
    For $x\in[\ulz,\olz]$ (mean-value, see~\cite[Section 2.4.3]{LJ-MK-OD-EW:01}),
    \begin{align*}
        Tf^\sfc(T^{-1}y,w) &\in T[\sfJ_x](T^{-1}y - \overcirc{x}) + T[\sfJ_w](w - \overcirc{w}) \\ 
        & \quad + T[\sfJ_u](N(T^{-1}y) - \overcirc{u}) + Tf(\overcirc{x},\overcirc{u},\overcirc{w}).
    \end{align*}
    Considering any interval $[\uly,\oly]\subseteq\Phi([\ulz,\olz])$ containing $y$, using $(C'_{[\uly,\oly]},\uld'_{[\uly,\oly]},\old'_{[\uly,\oly]})$ from Assumption~\ref{asm:linearapprox} on $N'$,
    \begin{align*}
        g^\sfc(y,w) &\in T[\sfJ_x](T^{-1}y - \overcirc{x}) + T[\sfJ_w](w - \overcirc{w}) \\ 
        & \hspace{-3em} + T[\sfJ_u](C'_{[\uly,\oly]}y + [\uld'_{[\uly,\oly]},\old'_{[\uly,\oly]}] - \overcirc{u}) + Tf(\overcirc{x},\overcirc{u},\overcirc{w}), \\
        g^\sfc(y,w) &\in T([\sfJ_x] + [\sfJ_u] (C'_{[\uly,\oly]}T))T^{-1}[\uly,\oly] \\ 
        & \hspace{-1em} + T[\sfJ_u] [\uld'_{[\uly,\oly]},\old'_{[\uly,\oly]}] + T[\sfR_{[\ulz,\olz]}(\ulw,\olw)].
    \end{align*}

    \vspace{-1.85em}
}
\end{proof}

It is important to note that \change{the inclusion functions $\sfJ_x,\sfJ_u,\sfJ_w$} in~\eqref{eq:Tclinclfun} are \change{evaluated} in the original coordinates. 
\change{Instead, one could symbolically} write $g$ as a new system and directly apply the closed-loop Jacobian-based inclusion function from~\eqref{eq:clinclfun}. In practice, however, these transformed dynamics often have complicated expressions that lead to excessive conservatism when using natural inclusion functions, and are not suitable for characterizing invariant sets.  %
In the next Theorem, we link forward invariant hyper-rectangles in transformed coordinates to forward invariant paralleletopes in standard coordinates.

\begin{theorem}[Forward invariance in paralleletopes] \label{thm:Tinvsets}
    Consider the neural network controlled system~\eqref{eq:cldynsys} \change{with the disturbance set $\calW=[\ulw,\olw]$. Let} $T\in\R^{n\times n}$ be an invertible matrix transforming the system into~\eqref{eq:Tclsys}\change{, with initial condition $y_0\in[\uly_0,\oly_0]$}. \change{Given a set $\calS\supseteq[\uly_0,\oly_0]$,} let $\sfG^\sfc_\calS$ be a $\calS$-localized inclusion function for $g^\sfc$, \emph{e.g.}~\eqref{eq:Tclinclfun}, and let $\sfE^{\change{\sfc}}_{T,\calS}$ be the embedding system~\eqref{eq:embsys} induced by $\sfG^\sfc_{\calS}$. If
    \[
        \sfE^{\change{\sfc}}_{T,\calS}(\uly_0,\oly_0,\ulw,\olw) \geqse 0,
    \]
    then the paralleletope $\Phi^{-1}([\uly_0,\oly_0])$ is a $[\ulw,\olw]$-robustly forward invariant set for the neural network controlled system~\eqref{eq:cldynsys}.
\end{theorem}

\change{
\begin{proof}
    Consider any trajectory $t\mapsto x(t)$ of the original system~\eqref{eq:cldynsys} starting from $x_0\in \Phi^{-1}([\uly_0,\oly_0])$. Given the one-to-one correspondence between~\eqref{eq:cldynsys} and~\eqref{eq:Tclsys}, the curve $t\mapsto \Phi(x(t))$ is the trajectory of the transformed system~\eqref{eq:Tclsys} starting from $y_0= \Phi(x_0)$. 
    Using Proposition~\ref{prop:invariance}, $\sfE^{\change{\sfc}}_{T,\mathcal{S}}(\uly_0,\oly_0,\ulw,\olw) \geqse 0$ implies that the hyper-rectangle $[\uly_0,\oly_0]$ is $[\ulw,\olw]$-robustly forward invariant in the transformed system~\eqref{eq:Tclsys}. This implies that $y(t)\in[\uly_0,\oly_0]$ for every $t\geq 0$. Since $x(t) = \Phi^{-1}(y(t))$, it follows that $x(t)\in\Phi^{-1}([\uly_0,\oly_0])$ for every $t\geq 0$.
\end{proof}
}
When using~\eqref{eq:Tclinclfun} as $\sfG^\sfc_\calS$, the neural network verification step from Assumption~\ref{asm:linearapprox} to find $(C',\uld',\old')$ is evaluated separately on each face of the hyperrectangle \change{$[\uly_0,\oly_0]$}, \emph{i.e.}, each face of the paralleletope \change{$\Phi^{-1}([\uly_0,\oly_0])$}. 
The dynamical approach from Theorem~\ref{thm:conv} yields a nested family of invariant hyperrectangles for the transformed system~\eqref{eq:Tclsys}, corresponding to a nested family of invariant paralleletopes for the original system~\eqref{eq:cldynsys}.
There are principled ways of choosing $T$\change{, \emph{e.g.},} the Jordan decomposition of the linearization about an equilibrium.
\section{Numerical Experiments} \label{sec:numerical}

Consider two vehicles $\mathrm{L}$ and $\mathrm{F}$ each with dynamics
\begin{align} \label{eq:LFdynamics}
    \begin{aligned}
        \dot{p}_x^j &= v_x^j, & \quad\quad \dot{v}_x^j &= \sigma(u_x^j) + w_x^j, \\
        \dot{p}_y^j &= v_y^j, & \quad\quad \dot{v}_y^j &= \sigma(u_y^j) + w_y^j,
    \end{aligned}
\end{align}
for $j\in\{\mathrm{L},\mathrm{F}\}$, where $p^j = (p_x^j,p_y^j)\in\R^2$ is the displacement of the center of mass of $j$ in the plane, $v^j = (v_x^j,v_y^j)\in\R^2$ is the velocity of the center of mass of $j$, $(u_x^j,u_y^j)\in\R^2$ are desired acceleration inputs limited by the nonlinear softmax operator $\sigma(u)=u_{\text{lim}}\tanh(u/u_{\text{lim}})$ with $u_\text{lim}=20$, and \change{$w_x^j,w_y^j\in [-0.005,0.005]$} are bounded disturbances on $j$. 
Denote the combined state of the system $x:=(p^\mathrm{L},v^\mathrm{L},p^\mathrm{F},v^\mathrm{F})\in\R^8$. We consider a leader-follower structure for the system, where the leader vehicle $\mathrm{L}$ chooses its control $u=(u_x^\mathrm{L},u_y^\mathrm{L})$ as the output of a \change{state feedback continuously applied} neural network \change{controller} ($4\times 100 \times 100 \times 2$, ReLU activations)\change{, with input $(p_x^\mathrm{L}, p_y^{\mathrm{L}}, v_x^\mathrm{L}, v_y^\mathrm{L})$.} The neural network was trained \change{using imitation learning on} $5.7$M data points from an offline MPC control policy for the leader only, with control limits implemented as hard constraints rather than $\sigma$. The offline policy minimized a quadratic cost aiming to stabilize to the origin while avoiding four circular obstacles centered at $(\pm 4, \pm4)$ with radius $2.25$ each, implemented as hard constraints with $33\%$ padding and a slack variable. 
The follower vehicle $\mathrm{F}$ follows the leader with a PD controller
\begin{align}
    u_{\bfd}^\mathrm{F} = k_{p} (p_\bfd^\mathrm{L} - p_\bfd^\mathrm{F}) + k_{v} (v_\bfd^\mathrm{L} - v_\bfd^\mathrm{F}),
\end{align}
for each $\bfd\in\{x,y\}$ with $k_p=6$ and $k_v=7$. 

\begin{figure}
    \centering
    \includegraphics[width=\columnwidth]{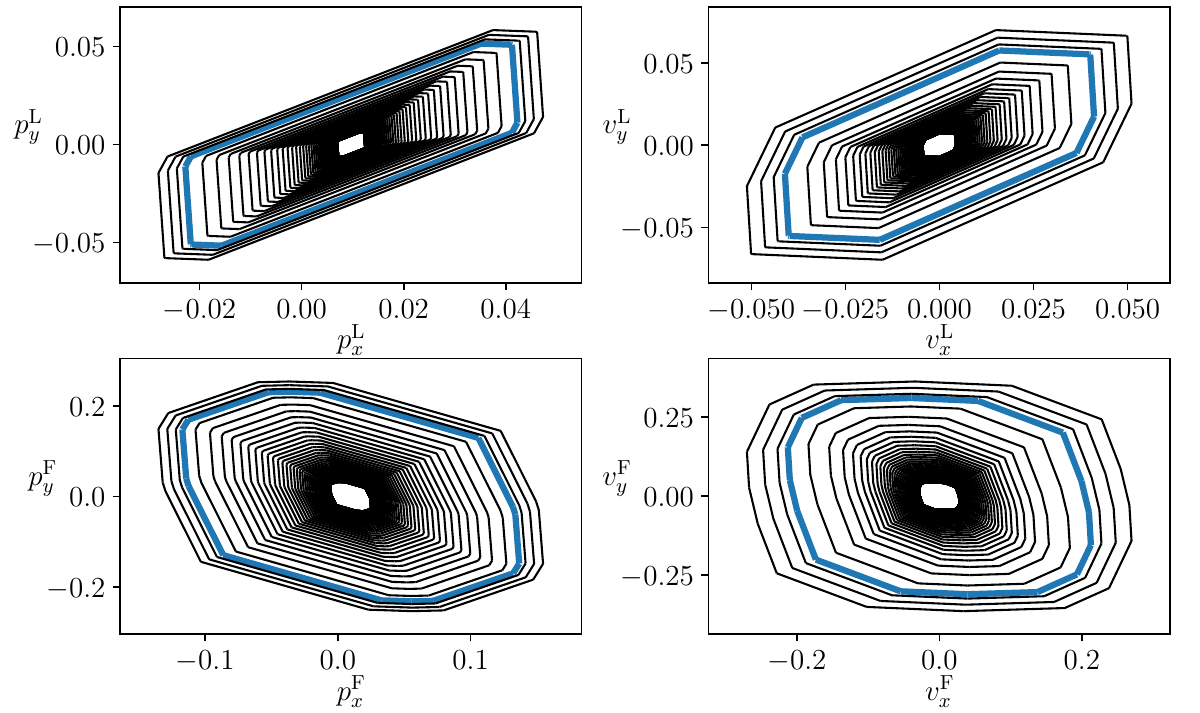}
    \caption{
    A family of \change{93} paralleletope invariant sets for the leader-follower system~\eqref{eq:LFdynamics} are visualized using projections from $\R^8$ onto $4$  $\R^2$ planes. 
    \textbf{Top Left:} projection onto leader's position $p_x^\mathrm{L},p_y^\mathrm{L}$; \textbf{Top Right:} projection onto leader's velocity $v_x^\mathrm{L},v_y^\mathrm{L}$; \textbf{Bottom Left:} projection onto follower's position $p_x^\mathrm{F},p_y^\mathrm{F}$; \textbf{Bottom Right:} projection onto follower's velocity $v_x^\mathrm{F},v_y^\mathrm{F}$. 
    After one invariant set $\Phi^{-1}([\uly_0,\oly_0])$ is found (blue line), the transformed embedding system is integrated forwards until \change{approximate} convergence, and backwards until Proposition~\ref{prop:invariance} is violated, yielding a monotonically decreasing collection of nested invariant sets converging to an attractive set (innermost lines).
    }
    \label{fig:leaderfollower}
    \vspace{-1em}
\end{figure}

First, a trajectory of the undisturbed system is run until it reaches the equilibrium point $x^\star\approx[0.01,0,0,0,0.01,0,0,0]^T$. Then, using \texttt{sympy}, the Jacobian matrices $D_\bfs f(x^\star,N(x^\star),0)$ for $\bfs\in\{x,u,w\}$ are computed, along with CROWN (\texttt{same-slope}) using \verb|auto_LiRPA|~\cite{xu2020automatic} along the interval $[\ulz,\olz] := x^\star + [-0.06,0.06]^4\times [-0.25,0.25]^2\times[-0.325,0.325]^2$ yielding $(C_{[\ulz,\olz]},\uld_{[\ulz,\olz]},\old_{[\ulz,\olz]})$. 
Then, the Jordan decomposition $T^{-1}JT = (\change{D_xf}(x^\star,N(x^\star),0) + \change{D_uf}(x^\star,N(x^\star),0)C_{[\ulz,\olz]})$ yields the transformation $T$, and the matrix $J\approx\operatorname{diag}(-6,-6,-4.12,-4.26,-0.93,-0.95,-1,-1)$ is filled with negative real eigenvalues, signifying that the equilibrium $x^\star$ is locally stable. 
The $T$-transformed system~\eqref{eq:Tclsys} is analyzed with the embedding system induced by~\eqref{eq:Tclinclfun}\change{, using \texttt{npinterval}~\cite{AH-SJ-SC:23b} to compute the natural inclusion functions of the symbolic Jacobian matrices to obtain $\sfJ_x,\sfJ_u,\sfJ_w$}. 
The interval $[\uly_0,\oly_0] = Tx^\star + [-0.05,0.05]^4\times[-0.08,0.08]^2\times[-0.11,0.11]^2$ yields $\Phi^{-1}([\uly_0,\oly_0])\subseteq[\ulz,\olz]$, and $\sfE_{T\change{,\Phi([\ulz,\olz])}}(\uly_0,\oly_0,\ulw,\olw) \geqse 0$. 
Thus, using Theorem~\ref{thm:Tinvsets}, the paralleletope $\Phi^{-1}([\uly_0,\oly_0])$ is an invariant set of~\eqref{eq:LFdynamics}. The embedding system in $y$ coordinates is simulated forwards using Euler integration with a step-size of $0.1$ for $90$ time steps, and at each step the localization $[\ulz,\olz]=T^{-1}[\uly(t),\oly(t)]$ is refined. 
Starting from $\smallconc{\uly_0}{\oly_0}$, the forward-time embedding system converges to a point $\smallconc{\uly^\star}{\oly^\star}$, where $\sfE_{T\change{,\Phi([\ulz,\olz])}}(\uly^\star,\oly^\star,\ulw,\olw)=0$. The transformed embedding system is also simulated backwards using Euler integration with a step-size of $0.05$ \change{until} the condition $\smallconc{\uly(t)}{\oly(t)}\geqse 0$ is \change{violated} (call the final time $t'$). 
Using Theorems~\ref{thm:conv}~and~\ref{thm:Tinvsets}, the collection $\{\Phi^{-1}([\uly(t),\oly(t)])\}_{t\geq t'}$ consists of nested invariant paralleletopes converging to the attractive set $\Phi^{-1}([\uly^\star,\oly^\star])$ with region of attraction $\Phi^{-1}([\uly(t'),\oly(t')])$. 
The initial paralleletope takes $0.38$ seconds to verify, and the entire nested family \change{of $93$ paralleletopes} takes $40.28$ seconds to compute.

\section{Conclusions}
Using interval analysis and neural network verification tools, we propose a framework for certifying hyper-rectangle and paralleletope invariant sets in neural network controlled systems. 
The key component of our approach is the dynamical embedding system, whose trajectories can be used to construct a nested family of invariant sets. 
This work opens up an avenue for future work in designing safe learning-enabled controllers.

\addtolength{\textheight}{-12cm}   %

\bibliographystyle{IEEEtran}
\bibliography{SJ,SJ2}

\end{document}